\documentclass[a4paper,USenglish,cleveref, autoref, thm-restate]{lipics-v2021}

\hideLIPIcs  

\usepackage[linesnumbered, lined, boxed, commentsnumbered]{algorithm2e}
\usepackage{xpatch}
\xpretocmd{\algorithm}{\hsize=\linewidth}{}{}
\usepackage{tablefootnote}

\usepackage{todonotes}

\usepackage{hyperref}
\usepackage{amsmath}
\usepackage{amssymb}
\usepackage{amsthm}
\usepackage{stmaryrd} 
\usepackage{mathpartir}
\usepackage{mathabx}
\usepackage{adjustbox}
\usepackage{mathtools}
\usepackage{caption}
\usepackage{tikz}
\usetikzlibrary{backgrounds, automata, positioning,calc, shapes.geometric, arrows, cd, intersections, patterns.meta}
\usepackage{amscd}
\usepackage{array}
\usepackage{rotating}
\usepackage{float}
\usepackage{mathtools}
\usepackage{wrapfig}
\usepackage{hyperref}
\DeclarePairedDelimiter{\ceil}{\lceil}{\rceil}

\makeatletter
\newbox\xrat@below
\newbox\xrat@above
\newcommand{\xrightarrowtail}[2][]{%
  \setbox\xrat@below=\hbox{\ensuremath{\scriptstyle #1}}%
  \setbox\xrat@above=\hbox{\ensuremath{\scriptstyle #2}}%
  \pgfmathsetlengthmacro{\xrat@len}{max(\wd\xrat@below,\wd\xrat@above)+.6em}%
  \mathrel{\tikz [>->,baseline=-.58ex,line width=0.43pt]
                 \draw (0,0) -- node[below=-2pt] {\box\xrat@below}
                                node[above=-2pt] {\box\xrat@above}
                       (\xrat@len,0) ;}}
\makeatother
\newcommand*{\move}[1]{\mathrel{\smash{\xrightarrowtail{\scriptscriptstyle#1}}}}%

\newcommand{\literal}[1]{\mathsf{#1}}
\newcommand{\E}{\mathcal{E}}
\newcommand{\N}{\mathbb{N}}
\newcommand{\Z}{\mathbb{Z}}
\newcommand*{\Min}{\operatorname{\mathrm{Min}}}

\newcommand{\w}{\text{wdh}}
\newcommand{\h}{\text{hgt}}
\newcommand*{\relSize}[1]{\lvert\mathord{#1}\rvert}
\newcommand{\winamin}{\wina^{\scriptscriptstyle\min}}
\newcommand{\wina}{\mathsf{Win}_a}
\newcommand{\pareto}{\mathsf{Pareto}_\mathcal{G}}
\newcommand{\iteration}{\mathsf{Iteration}_\mathcal{G}}
\newcommand*{\dom}{\operatorname{\mathrm{dom}}}

\definecolor{blue1}{RGB}{0,62,107}
\definecolor{orange1}{RGB}{255, 183, 2}
\definecolor{orange2}{RGB}{243,145,0}
\definecolor{blue2}{RGB}{0,171,217}
\definecolor{blue3}{RGB}{161,217,248}

\title{Galois Energy Games}
\subtitle{To Solve All Kinds of Quantitative Reachability Problems}

\author{Caroline Lemke}{
Department of Computing Science,
Carl von Ossietzky Universität Oldenburg,
Germany}{caroline.lemke@uni-oldenburg.de}{https://orcid.org/0009-0008-7863-0818}{}

\author{Benjamin Bisping}{
Institute of Software Engineering and Theoretical Computer Science,
Technische Universität Berlin,
Germany \and
\url{https://bbisping.de}}{benjamin.bisping@tu-berlin.de}{https://orcid.org/0000-0002-0637-0171}{}

\authorrunning{C. Lemke and B. Bisping} 

\Copyright{Caroline Lemke and Benjamin Bisping} 



\category{} 

\relatedversion{} 

\acknowledgements{We want to thank Gabriel Vogel for the WebGPU prototype implementation and Uli Fahrenberg and Lukas Panneke for their feedback.}

\nolinenumbers 

\begin{document}

\maketitle

\begin{abstract}
We provide a generic decision procedure for energy games with energy-bounded attacker and reachability objective, moving beyond vector-valued energies and vector-addition updates. 
All we demand is that energies form well-founded bounded join-semilattices, and that energy updates have an upward-closed domain and can be “undone” through a Galois-connected function. 
We instantiate these \emph{Galois energy games} to common energy games, declining energy games, multi-weighted reachability games, coverability on vector addition systems with states and shortest path problems, supported by an Isabelle-formalization and two implementations. 
For these instantiations, our simple algorithm is polynomial w.r.t.\ game graph size and exponential w.r.t.\ dimension.
\end{abstract}

\section{Introduction}

Most problems of how to reach a goal at minimal cost can be expressed as \emph{energy games}~\cite{generalized-energy-games, reachability-games-and-friends}.
In energy games, players try to achieve objectives, but fail if they run out of resources that they can gain or lose on the way.
But strangely, existing literature does not quite explain how to solve the following kind of multi-weighted energy game:

\begin{figure}[h]
    \centering    
    \begin{subfigure}[b]{0.5\textwidth}
    \begin{adjustbox}{center, trim={0mm 0mm 18mm 0mm}}
    \adjustbox{max height=5.5cm}{
        \centering
         \begin{tikzpicture}[>->, shorten <=1pt, shorten >=0.5pt, auto, node distance=2cm,
              posStyle/.style={draw, inner sep=1ex, minimum size=1cm, minimum width=2.25cm, anchor=center, draw, black, fill=gray!5},
              defender/.style={ellipse, inner sep=0ex}]
            
              \node[posStyle, initial, initial text={}]
                (Office){$\literal{Office}$};
              \node[posStyle]
                (CM) [above = 1.5cm of Office] {$\literal{CoffeeMaker}$};
              \node[posStyle, defender]
                (DH) [right = 2.5cm of CM] {$\literal{Department Head}$};
              \node[posStyle, defender, label={0:\checkmark}]
                (Fin) [right = 3cm of Office] {$\literal{Energized}$};
            
              \path
                (Office)
                  edge [out=285,in=255,looseness=8] node {$-1\,\literal{Shots}, +1\,\literal{E}$} (Office)
                  edge[bend left = 10] node {} (CM)
                  edge[bend right = 5] node {$-10\,\literal{E}$} (Fin)
                (CM)
                  edge [out=105,in=75,looseness=8] node {$\literal{Shots} \leftarrow \min\{\literal{Shots} + 1, \literal{Cups}\}$} (CM)
                  edge[bend left = 10] node {$-2\,\literal{T}$} (Office)
                  edge[bend left = 10] node {} (DH)
                (DH)
                  edge[bend right = 15] node {$-1\,\literal{T}$} (Office)
                  edge[bend left = 15, pos=.3] node {$-1\,\literal{Shots}, -1\,\literal{Cups}$} (Office)
              ;
            \end{tikzpicture}
    }
    \end{adjustbox}
         \caption{Game graph.}
         \label{fig:coffee-graph}
    \end{subfigure}
    \hfill
    \begin{subfigure}[b]{0.3\textwidth}
    \begin{adjustbox}{center}
    \adjustbox{max height=5.5cm}{
        \centering
        \begin{tikzpicture}[]
          \begin{scope}[scale=.26]
            \draw[help lines, opacity=.4] (0,0) grid (10.5,20.5);
            \draw[thick, red, opacity=0.7, text opacity=.8]
                (1,21) -- (1,20) -- (2,20) -- (2,10) -- (3,10) -- (3,6) -- (4,6) -- (4,4) -- (5,4) -- (5,2) -- (10, 2) -- (10,1) -- (11,1);
            \node[opacity=.4] at (12,0){$\mathsf{Cups}$};
            \node[opacity=.4] at (-0.5,21.3){$\mathsf{T}$};
            \foreach \x in {2,4,...,10}
            {
                \node[opacity=.4] at (\x,-0.7){\x};
            }
            \foreach \x in {2,4,...,20}
            {
                \node[opacity=.4] at (-0.8,\x){\x};
            }
          \end{scope}
        \end{tikzpicture}
        }
        \end{adjustbox}
        \caption{Minimal cost of time/cups.}
        \label{fig:coffee-pareto}
    \end{subfigure}
    \caption{Espresso energy game of \autoref{exm:running-example}.}
    \label{fig:coffee-example}
\end{figure}
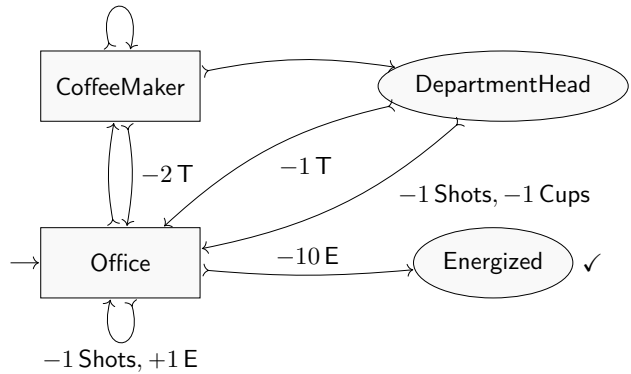
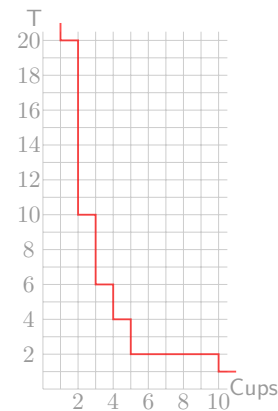

\begin{example}
    \label{exm:running-example}
    The energy game of \autoref{fig:coffee-graph} models the situation of a researcher, who hosts a coauthor at his office and wants to energize her with ten shots of espresso.
    For this, he leaves his $\literal{Office}$ for the $\literal{CoffeeMaker}$ and brews some $\literal{Shots}$ (bounded by the number of $\literal{Cups}$ he carries).
    Upon returning to the $\literal{Office}$, the shots can be transformed into energization $\literal{E}$.
    The return usually takes two units of time $\literal{T}$.
    But there is a quicker way, passing by the office of the $\literal{DepartmentHead}$.
    Unfortunately, the department head will either cost our host some time by chatting or help herself to one of the filled coffee cups. (Which of the two happens is not under his control.)
    Depending on the number of cups, he has to take several rounds until energization $\literal{E}$ reaches level $10$ and the game can be won by moving to $\literal{Energized}$.

    Assuming that we start with no energization and shots ($\literal{E} = \literal{Shots} = 0$), but with a budget of cups and time, how much of the two does the host need to reach his objective?
    There are several answers, summarized by the Pareto front of cups and time $\{(1,20), (2,10), (3,6), (4, 4), (5, 2), (10, 1)\}$, which is drawn in \autoref{fig:coffee-pareto}.

    The strategy of the host depends on the resources:
    If there is only one cup, then he should never pass by the department head, where it might be lost.
    So, he will cycle through $\literal{Office} \move{\quad} \literal{CoffeeMaker} \move{\literal{Shots} \leftarrow \min\{\literal{Shots} + 1, \literal{Cups}\}} \literal{CoffeeMaker} \move{-2\,\literal{T}} \literal{Office} \move{-1\,\literal{Shots}, +1\,\literal{E}} \literal{Office} \dots$ until the espressos have been served after at least $10 \cdot 2\,\literal{T}$ time units.
    If the host brings 10 or more cups, he can fill them all and then safely take the quicker route.
    There might be a delay at the head's office or it might take a second round, but $1\,\literal{T}$ will suffice.
    For the scenarios in-between, things become more involved \dots would it not be nice to have an algorithm to find the minimal budgets there?
\end{example}

\noindent
\emph{Energy games} are often studied in light of their closeness to mean-payoff games and (games on) vector addition systems with states~\cite{abdulla-parity-on-integer, reachability-games-and-friends, generalized-energy-games}.
They can model how different resources may interact with each other, for example, in a trade-off between one resource gained (energization) and another lost (espresso shots).
Energy updates are usually understood as vector addition.
But vector addition is too restricted to capture \autoref{exm:running-example}, where the $\literal{Cups}$ component bounds the $\literal{Shots}$ component!

A recent article on process equivalences introduces \emph{declining energy games}, which allow resources to constrain one-another~\cite{bens-algo}. But these games prohibit increasing updates!
Another recent article formally studies decision procedures and bounds for \emph{multi-weighted reachability games}~\cite{mult-reach-games}, which roughly are the intersection of declining and standard energy games, that is, even more restricted.

With this work,
we show that we can have all this cake and eat it:
Positive as well as negative updates (even multiplicative ones), components bounding each other, and a formalized solution.
All we need is that energy updates can be undone through \emph{Galois connections}, a weak form of inversion on monotonic functions.
Other than that, we demand little of the precise objects used as energies, which makes our results quite versatile.

\subparagraph*{Contributions}

\begin{itemize}
    \item \autoref{sec:energygames} introduces \emph{Galois energy games} as a generalization of existing energy game models with an \emph{inductive characterization of winning budgets}.
    \item \autoref{sec:decidability} provides a \emph{simple fixed-point algorithm to compute winning budgets} of Galois energy games and proves its correctness - formalized in Isabelle/HOL. 
    \item We examine how other energy game models and quantitative reachability problems can be \emph{understood as instances of our result} in \autoref{sec:instances} and discuss two prototype implementations.
    \item We give precise \emph{parametric upper bounds for the complexity} of the algorithm and its instantiations in \autoref{sec:complexity}: polynomial in game graph size, exponential in dimensionality.
    \item In \autoref{sec:winning-conditions}, we discuss how other scenarios and winning conditions could be addressed.
\end{itemize}

\section{Energy Games}\label{sec:energygames}

Energy games are zero-sum two-player games with perfect information, played on directed graphs labeled by energy update functions. Energies represent resources such as time or espresso shots, often modeled as vectors.
In our quest for genericity, we model energies as arbitrary partially ordered sets, where each subset has at least one minimal element and suprema of finite subsets exist, i.e.\ well-founded bounded join-semilattices.
The energy updates describe how each move across an edge affects the current energy, altering the resources as players navigate the game.

\subsection{Generic Energy Games}

We introduce energy games as generically as possible.
For this, we first fix a set of energies and assume $(\E, \leq)$ to be a well-founded bounded join-semilattice.
Then, we can define energy games played on directed graphs labeled by updates.

\begin{definition}[Energy game]\label{def:EnergyGames} 
    An \emph{energy game} $\mathcal{G} = (G,G_a, \move{ \ })$ over $(\E, \leq)$ consists of
    \begin{itemize}
        \item a set of positions $G$, partitioned into attacker positions $G_a \subseteq G$ and defender positions $G_d := G \setminus G_a$, and
        \item a partial function $\move{ \ } \ : (G \times G) \rightarrow (\E \rightarrow \E)$ mapping edges to partial energy functions.
    \end{itemize}
    For $g,g' \in G$ we write $g \move{u} g'$ instead of $\move{ \ } (g,g') = u$ and call $u$ an update in $\mathcal{G}$. The game with initial position $g_0 \in G$ and initial energy $e_0 \in \E$ is denoted as $\mathcal{G}[g_0, e_0]$.
\end{definition}

\begin{example}
    Formally, \autoref{def:EnergyGames} does not permit the two edges between $\mathsf{DepartmentHead}$ and $\mathsf{Office}$ in \autoref{exm:running-example}.
    However, this can be fixed by adding a node $\mathsf{Chat}$ such that $\mathsf{DepartmentHead} \move{-1 \, \mathsf{T}} \mathsf{Chat} \move{$\quad$} \mathsf{Office}$.
    In the example, the researcher getting coffee is the attacker, choosing at rectangular positions, while the environment is the defender at elliptical positions.
    The resources (cups, time, shots, energization) can be represented as four-dimensional natural vectors, making \autoref{exm:running-example} an energy game over $\N^4$ with the component-wise order.
\end{example}
In an energy game, the players move one token from one position to the next according to the edges, resulting in a play. At attacker positions, the attacker chooses (according to some attacker strategy) and vice versa.  

\begin{definition}[Play, strategy]
    Let $\mathcal{G} = (G,G_a, \move{ \ })$ be an energy game over $(\E, \leq)$. 
    A \emph{play} in $\mathcal{G}$ is a possibly infinite walk $\rho = g_0 g_1 ... \in G^* \cup G^\omega$ in the underlying directed graph.
    
    Let $G_p \in \lbrace G_a , G_d \rbrace $. A \emph{strategy} $s: G^* G_p \rightarrow G$ is a partial function mapping finite plays to a next position such that $s(g_0 ... g_n)$ is some successor of $g_n$ if $g_n \in G_p$ is not a deadlock. If $G_p=G_a$ (resp. $G_p = G_d$) we call $s$ an attacker (resp. defender) strategy.
    A play $\rho = g_0 g_1 ...$ is consistent with $s$ if $g_i \in G_p$ implies $g_{i+1} = s(g_0 ... g_i)$ for all $i + 1< |\rho |$.
\end{definition}
During a play, an initial energy is repeatedly updated according to the edges in the play. This is defined as the energy level, which keeps track of the resources. The attacker wins a play if and only if it ends in a deadlocked defender position before running out of energy, i.e.\ without the energy level becoming undefined.

\begin{definition}[Energy level, winning]\label{def:winning}
    Let $\mathcal{G} = (G,G_a, \move{ \ })$ be an energy game over $(\E, \leq)$.  
    The \emph{energy level} of a play $\rho= g_0 g_1 \ldots$ w.r.t.\ initial energy $e_0 \in \E$ is inductively defined as $EL(\rho, e_0, 0) := e_0$ and $EL(\rho, e_0, i+1) := u (EL(\rho, e_0, i))$ where $g_i \move{u} g_{i+1}$ and $i+1 < |\rho|$.
    
    If the energy level becomes undefined or the play is infinite, the defender wins the play. 
    If the last position $g_l$ of a finite play with the final energy level $EL(\rho, e_0, | \rho | -1 )$  defined is a deadlock, then the defender wins the play if $g_l \in G_a$ and the attacker wins if $g_l \in G_d$.
    
    For $g_0\in G$ and $e_0 \in \E$ an attacker winning strategy for $\mathcal{G}[g_0, e_0]$ is an attacker strategy ensuring the attacker to win all consistent plays. The attacker wins the game starting with $e_0$ from $g_0$ if and only if such an attacker winning strategy exists.
    The \emph{attacker winning budget} of $g\in G$ is defined as ${\wina(g) = \lbrace e \in \E \ | \ \exists s. \ s \text{ is an attacker winning strategy for } \mathcal{G}[g, e]\rbrace}$. 
\end{definition}
In other settings, infinite plays (where the energy level is always defined) are won by the attacker~\cite{generalized-energy-games, automata-bounds}. We will revisit this distinction in \autoref{sec:winning-conditions} as safety winning conditions.

\begin{remark}\label{rem:reachability}
The winning conditions in \autoref{def:winning} encode quantitative reachability problems: In reachability games~\cite{reachability-games-and-friends}, the attacker aims to reach a fixed target set in a (weighted) game graph. By treating defender deadlock positions as the target set, we formulated (quantitative) reachability winning conditions.

When starting with a multi-weighted reachability game~\cite{mult-reach-games}, i.e.\ a game $(G,G_a, \move{ \ })$ weighted by vectors of naturals with a target set $F \subseteq G$, this can easily be transformed to an energy game:
Remove outgoing edges of positions in $G_d  \cap F$ and add a loop to every deadlock in $G_d \setminus F$. 
Further, add a deadlocked defender position with an ingoing edge labeled by the zero-vector from each position in $G_a \cap F$. By this construction, a defender positions is a deadlock if and only if it corresponds to a position in $F$ in the original game.
In the reachability game one is interested in the minimal ensured cost\footnote{Note that Brihaye and Goeminne~\cite{mult-reach-games} refer to energies as costs. While this is fitting for settings where only vector addition is considered, we want to note that in energy games some resources might be necessary to have in order to win without ever being used up -- like at least one cup in \autoref{exm:running-example}.}
for a position $g$, i.e.\ a minimal upper bound the attacker can enforce while winning for the sum of edge weights of plays starting in $g$.
By switching from addition to subtraction this corresponds to asking for the minimal attacker winning budget of $g$ in the constructed energy game.
\end{remark}
\smallskip
Defender winning strategies and budgets can be defined symmetrically to \autoref{def:winning}. However, it is sufficient to focus on the attacker's objective, since energy games are determined w.r.t.\ a starting position and initial energy. This can be shown by constructing a parity game, which is positionally determined. For more details we refer to the appendix.

\begin{lemma}\label{lem:pos-determinacy}
    Let $\mathcal{G}$ be an energy game over $(\E, \leq)$.
    For all positions $g\in G$ and energies $e \in \E$ either the defender or the attacker wins $\mathcal{G}[g,e]$. Then, there exists an energy-positional winning strategy, i.e.\ a winning strategy that depends only on the current position and energy level. 
\end{lemma}

\subsection{Inductive Characterization of Attacker Winning Budgets}
\label{subsec:inductive-attacker-win}

We want to decide energy games by calculating the attacker winning budgets. For this we study attacker winning budgets and observe that they are upward-closed in monotonic energy games, i.e.\ where all updates are monotonic and have an upward-closed domain.

\begin{definition}[Monotonic energy game]
    Let $\mathcal{G}$ be an energy game over $(\E, \leq)$. Then, $\mathcal{G}$ is a \emph{monotonic energy game} if for all updates $u$ in $\mathcal{G}$ and all $e,e'\in \E$ with $e\leq e'$ the following implication holds: If $u(e)$ is defined then so is $u(e')$ and $u(e) \leq u(e')$. 
\end{definition}

\begin{lemma}\label{lem:upwards-closure}
    Let $\mathcal{G}$ be a monotonic energy game over $(\E, \leq)$.
    Then, the attacker winning budgets are upward-closed, i.e.\ 
    $\wina (g) = \ \uparrow \wina (g) := \{ e \in \E \ | \ \exists e' \in \wina (g). \ e' \leq e \}$ for all positions $g\in G$.
\end{lemma}
\autoref{lem:upwards-closure} follows directly from the definition of winning budgets and monotonic games since the same strategy serves as a witness when starting with more energy. 

In monotonic energy games, the attacker winning budgets can be characterized by finite Pareto fronts, i.e.\ by minimal attacker winning budgets $\winamin (g) : = \Min \wina (g)$. (For a set $M$ we distinguish between the set of all minimal elements, i.e.\ $\Min M := \lbrace m \in M \ | \ \forall m' \in M. \ m' \nleq m \rbrace$, and the minimum. We write $\min M$ to denote the minimum if it exists.) Since energies are well-founded, antichains and thereby the sets $\Min \wina (g)$ are finite. 

Attacker winning budgets have more structure, in particular, they can be characterized inductively.

\begin{lemma}\label{lem:inductive_wina}
    Let $\mathcal{G}$ be an energy game over $(\E, \leq)$. The attacker winning budgets can be characterized inductively as follows:
\begin{mathparpagebreakable}
    \inferrule{
      g \in G_a\\
      g \move{u} g'\\
      u(e) \in \wina(g')
    } {
      e \in \wina(g)
    }\and
    \inferrule{
      g \in G_d\\
      \forall g'. \  g \move{u} g' \longrightarrow u(e)\in \wina(g')
    } {
      e \in \wina(g)
    }
  \end{mathparpagebreakable}
\end{lemma}
\begin{proof}
An attacker winning strategy can be constructed from attacker winning strategies at successors w.r.t.\ the accordingly updated energies. This implies soundness. 

To argue completeness we utilize \autoref{lem:pos-determinacy}. 
Assuming $e \in \wina (g)$, there exists an energy-positional attacker winning strategy $s$ for $\mathcal{G}[g,e]$. 
We then fix an order induced by $s$ by setting $(g'',e'') \leq_s (g' ,e')$ if $(g'', e'') = (g',e')$ or there is an edge in the configuration graph from $(g',e')$ to $(g'', e'')$ consistent with $s$, i.e.\ if $g' \move{u} g''$, $e'' = u(e')$ and $g' \in G_a \longrightarrow s(g',e') = g''$.
Then, a derivation of $e \in \wina (g)$ can be found by well-founded induction.    
\end{proof}

\subsection{Galois Energy Games}

In order to convert the inductive characterization of \autoref{lem:inductive_wina} into a fixed-point algorithm, it seems that we need to \emph{invert} energy updates in the game.  However, it turns out that what is sufficient is a weaker form of invertibility given by \emph{Galois connections}. 

\begin{definition}[Galois energy game]\label{def:galois-games}
    A \emph{Galois energy game} is a monotonic energy game $\mathcal{G}$ over $(\E, \leq)$ such that $u^{\circlearrowleft}(e') := \min \lbrace e \in \E \ | \ e' \leq u(e) \rbrace $ exists and is computable for all updates $u$ in $\mathcal{G}$ and $e'\in \E$.
\end{definition}
Galois connections can be characterized by one of the following four equivalent properties. We refer to Erné et al.~\cite{galois} for more detail.

\begin{lemma}\label{lem:galoisproperties}
    Let $\mathcal{P}=(P, \leq_P)$ and $\mathcal{Q}=(Q, \leq_Q)$ be partially ordered sets with functions $f: P \rightarrow Q$ and $g: Q \rightarrow P$. 
    Then the following are equivalent:
\begin{enumerate}
    \item $f(p) \leq_Q q \longleftrightarrow p \leq_P g(q)$ holds for all $p \in P$ and $q \in Q$.
    \item $f$ and $g$ are monotonic, $g\circ f$ is increasing, and $f \circ g$ is decreasing. 
    \item $f$ is monotonic and $g(q)=\max_P \lbrace p \in P \ | \  f(p) \leq_Q q \rbrace$ for each $q\in Q$.\label{lem:galoisMax}
    \item $g$ is monotonic and $f(p) = \min_Q \lbrace q  \in Q \ | \ p \leq_P g(q) \rbrace$ for each $p\in P$.
\end{enumerate}

\end{lemma}
If any of the (and thereby all) properties hold, then $f$ and $g$ form a Galois connection between $P$ and $Q$. 
By \autoref{lem:galoisproperties} an energy game $\mathcal{G}$ over $(\E, \leq)$ is a Galois energy game if and only if each update $u$ in $\mathcal{G}$ has an upward-closed domain and there exists a computable function $u^{\circlearrowleft}: \E \rightarrow \dom(u)$ such that $u^{\circlearrowleft}$ and $u$ form a Galois connection between $\E$ and $\dom (u)$. Then, the minima used to define $u^\circlearrowleft$ in \autoref{def:galois-games} always exist. 

\begin{example}
    The energy game from \autoref{exm:running-example} is a Galois energy game. In particular, 
    $\left( \literal{Shots} \leftarrow \min\{\literal{Shots} + 1, \literal{Cups}\} \right)^{\circlearrowleft} (e_{\literal{C}},e_{\literal{T}},e_{\literal{S}},e_{\literal{E}}) = (\max\{e_{\literal{C}}, e_{\literal{S}}\}, e_{\literal{T}} , \max\{e_{\literal{S}} -1 , 0 \}, e_{\literal{E}})$.
\end{example}

\section{Deciding Galois Energy Games}
\label{sec:decidability}

For a fixed position $g$, the \emph{known initial credit problem} asks whether energy $e$ is sufficient for the attacker to always win when starting from $g$, i.e.\ whether $e \in \wina (g)$.
The \textit{unknown initial credit problem} asks, whether there exists an energy sufficient for the attacker to win from $g$, i.e.\ whether $\wina (g) \neq \varnothing$.
Subsuming both these problems, we provide an algorithm calculating the \emph{minimal attacker winning budgets of each position}.

The core idea of the algorithm is that of a shortest path algorithm:
We start at deadlocks and move backwards through the graph using the inductive characterization of attacker winning budgets (\autoref{lem:inductive_wina}).
To calculate energies backwards, we use Galois connections, allowing us to prove:

\begin{theorem}[Decidability of Galois energy games]\label{thm:decidable}
    Let $\mathcal{G}$ be a Galois energy game over $(\E, \leq)$ with a finite set of positions. 
    Then, the (un)known initial credit problem is decidable.
\end{theorem}

\subsection{The Algorithm}
To prove \autoref{thm:decidable}, we first discuss \autoref{alg:game-algorithm}. 
This is a simplified version of an algorithm proposed by Bisping~\cite{bens-algo}.

\LinesNumbered
\DontPrintSemicolon
\SetKwProg{For}{for}{\string:}{}
\SetKwProg{Fn}{def}{\string:}{}
\SetKwProg{While}{while}{\string:}{}
\SetKwIF{If}{ElseIf}{Else}{if}{\string:}{elif}{else\string:}{}%

\providecommand*{\code}[1]{\texttt{#1}}

\global\long\def\defEq{\mathrel{\coloneqq}}%
\global\long\def\codeStyle#1{\mathrm{#1}}%
\global\long\def\varname#1{\mathsf{#1}}%
\global\long\def\ccsRes#1{\left(\boldsymbol{\nu}#1\right)}%
\global\long\def\ccsStop{\mathbf{0}}%
\global\long\def\ccsPrefix{\ldotp\!}%
\global\long\def\ccsChoice{+}%
\global\long\def\ccs{\mathsf{CCS}}%
\global\long\def\ccsIdentifier#1{\mathsf{#1}}%
\global\long\def\ccsOutm#1{\overline{#1}}%
\global\long\def\ccsInm#1{#1}%

\global\long\def\rel#1{\mathcal{#1}}%
\global\long\def\bigo{\mathcal{O}}%
\newcommand{\powerSet}[1]{\mathbf{2}^{#1}}
\newcommand{\bellNumber}[1]{\mathbf{B}(#1)}
\newcommand{\compose}{\mathbin{\circ}}
\newcommand{\nats}{\mathbb{N}}
\newcommand{\ints}{\mathbb{Z}}
\newcommand{\domain}{\operatorname{\mathrm{dom}}}
\newcommand*{\vectorComponents}[2][n]{({#2}_1,\ldots,{#2}_{#1})}

\newcommand*{\gameMoveX}[1]{\mathrel{\smash{\xrightarrowtail{\scriptscriptstyle#1}}}}%
\newcommand*{\gameMove}{\gameMoveX{\hspace*{0.5em}}}%
\newcommand*{\game}{\mathcal{G}}%
\newcommand*{\gameSpectroscopy}{\mathcal{G}_\triangle}%
\newcommand*{\gameSpectroscopyClever}{\mathcal{G}_\blacktriangle}%
\newcommand*{\attackerPos}[2][]{{{\color{gray}(}#2{\color{gray})}}_\mathtt{a}^{\color{gray}\smash{\scriptscriptstyle#1}}}
\newcommand*{\defenderPos}[2][]{{{\color{gray}(}#2{\color{gray})}}_\mathtt{d}^{\color{gray}\smash{\scriptscriptstyle#1}}}
\newcommand*{\partition}[1]{\mathscr{#1}}
\newcommand*{\conjClosure}[1]{\lceil #1 \rceil^\land}
\newcommand*{\attackerSubscript}{{\operatorname{a}}}
\newcommand*{\defenderSubscript}{{\operatorname{d}}}
\newcommand*{\energies}{\mathbf{En}}
\newcommand*{\energyUpdates}{\mathbf{Up}}%
\newcommand*{\energyUpdate}{\mathsf{upd}}%
\newcommand*{\energyUpdateInv}{\mathsf{upd}^{-1}}%
\newcommand*{\updMin}[1]{\mathtt{min}_{\{\!#1\!\}}}
\newcommand*{\energyLevel}{\mathsf{EL}}%
\newcommand*{\attackerWin}{\mathsf{Win}_\attackerSubscript}
\newcommand*{\attackerWinMin}{\attackerWin^{\scriptscriptstyle\min}}
\newcommand*{\defenderWinMax}{\mathsf{Win}_\defenderSubscript^{\scriptscriptstyle\max}}

\begin{algorithm}[h]
  \Fn{$\varname{compute\_winning\_budgets}(\mathcal{G}=(G,G_a,\move{ \ }))$}{

    $\varname{win}\, := [g\mapsto \{\} \ | \ g\in G]\label{code:start-at-0}$


    \SetKwRepeat{Do}{do}{while}
    \Do{$\varname{win} \neq \varname{old\_win}$\label{code:end-while}}{\label{code:begin-while}

    $\varname{old\_win} := \varname{win}$
    
    \For{$\varname{g} \in G \ $\label{code:begin-iteration}}
    {
      \eIf{$\varname{g} \in G_a$}{

        $\varname{win}[\varname{g}] := 
          \Min \{ u^\circlearrowleft(e') \mid {\varname{g} \move{u} g'} \land \mathit{e'} \in \varname{old\_win}[g']\}$
        \label{code:attacker-pos-update}

      } {
        $\varname{win}[\varname{g}]  :=
           \Min \lbrace \sup \lbrace u^\circlearrowleft(e_{g'}) |  \varname{g} \move{u} g' \rbrace \ | \ \forall g'. \ \varname{g} \move{u} g' \longrightarrow e_{g'} \in \varname{old\_win}[g'] \rbrace$\label{code:defender-pos-update}
      }
    }
    }

    \KwRet{$\varname{old\_win}$}

  }
   \caption{\label{alg:game-algorithm}Computing minimal attacker winning budgets of Galois energy game $\mathcal{G}$.}
\end{algorithm}

Starting by assigning the empty set to each position in \autoref{code:start-at-0}, we then apply lines~\ref{code:begin-while}~to~\ref{code:end-while} adding sufficient energies to winning budgets and taking the minimal elements (which exist by well-foundedness) until we reach a fixed point to return. 
For a play to be won by the attacker, it has to end in a defender deadlock position. At such a position $g_d$ all energies are in the attacker winning budget and $\winamin (g_d) = \Min \E$. 
If a defender position has successors, then an energy is in the attacker winning budgets if it is sufficient to win moving to any successor the defender might choose. In Galois energy games, that is $e \in \wina (g_d)$ for $g_d\in G_d$ if $u^{\circlearrowleft} (e_{g'}) \leq e$ for some $e_{g'} \in \wina (g')$ for each $g'$ with $g_d \move{u} g'$.
This is ensured by \autoref{code:defender-pos-update}. (Note that the supremum always exists because the energies form a bounded join-semilattice and the sets in question are finite.)
Similarly, sufficient energies for attacker positions can be calculated from energies in the winning budget of successors, as done in \autoref{code:attacker-pos-update}. In particular, $g_a \in G_a$ and $e' \in \wina (g')$ for some $g'$ with $g_a \move{u} g'$ implies $u^{\circlearrowleft} (e') \in \wina (g_a)$.

\begin{table}[h]
    \centering
    \begin{tabular}{c|lllll}
        Iteration & 0 & 1 & 2 & 3 & $\ldots$ 
        \\
        \hline
        $\mathsf{Energized}$ & $\varnothing$ & $\lbrace (0,0,0,0) \rbrace$ & $\lbrace (0,0,0,0) \rbrace$  & $\lbrace (0,0,0,0) \rbrace$ & $\ldots$ 
        \\
        $\literal{Office}$ & $\varnothing$ & $\varnothing$ & $\lbrace (0,0,0,10) \rbrace$ & 
        $\lbrace (0,0,0,10), (0,0,1,9) \rbrace$&  $\ldots$ 
        \\
        $\mathsf{CoffeeMaker}$ & $\varnothing$ & $\varnothing$ & $\varnothing$& $\lbrace (0,2,0,10) \rbrace$ & $\ldots$ 
    \end{tabular}
    \caption{Application of \autoref{alg:game-algorithm} to \autoref{exm:running-example} \label{table:algorunningexm}}
\end{table}

\autoref{table:algorunningexm} illustrates how the information, which positions are winnable, travels through the graph when applying \autoref{alg:game-algorithm} to \autoref{exm:running-example}. The vectors represent cups, time, shots and energization in that order.

\subsection{Correctness}

Since we are considering a fixed point algorithm, we use Kleene's fixed point theorem~\cite{kleenes-fp-1981, kleene-fp-AFP} to prove \autoref{thm:decidable}. The following version relies on the notion of directed sets, i.e.\ a subset of a partially ordered set where all pairs of elements have an upper bound in the subset. In a directed-complete partial order all directed sets have a (directed) supremum. A Scott-continuous function preserves directed suprema.

\begin{lemma}[Kleene's fixed point theorem]\label{lem:kleene}
    Let $(P, \leq_P)$ be a directed-complete partial order with least element $p \in P$ and a Scott-continuous function $f: P \rightarrow P$.
    Then, the set of fixed points of $f$ in $P$ forms a complete lattice with least fixed point $\sup_P \lbrace f^i (p) \ | \ i \in \N \rbrace$.
\end{lemma}

\begin{proof}[Proof of \autoref{thm:decidable}]
    To utilize Kleene's fixed point theorem, we define a partial order on mappings from positions to possible Pareto fronts, i.e.\ antichains in $\E$. We call the set of such mappings $\pareto \subseteq (G \to 2^\E)$ and define 
    \begin{equation*}
            F \preceq F' \longleftrightarrow \forall g \in G. \  F(g) \subseteq \ \uparrow F'(g) \quad \text{ and } \quad \left(\textstyle{\sup_\preceq} P \right) (g) := \Min \smash{\bigcup_{F \in P}} \uparrow F(g)
    \end{equation*}
    for $F, F' \in \pareto$ and $P\subseteq \pareto$.
    These definitions yield that $(\pareto, \preceq)$ is a directed-complete partial order with minimal element $\mathbf{0}:= \lambda g. \ \varnothing$. 

    We then consider the function $\iteration : \pareto \rightarrow \pareto$ corresponding to one iteration of the while-loop in \autoref{alg:game-algorithm} when applied to $\mathcal{G}$ and show that it is Scott-continuous w.r.t.\ the $\preceq$-order, i.e.\ $\iteration (\sup_\preceq P) = \sup_\preceq \lbrace \iteration (F) \ | \ F \in P \rbrace$ for all directed sets $P \subseteq \pareto$. 
    
    Let $P \subseteq \pareto$ be such a directed set.
    The definition of $\iteration$ and monotonicity of all $u^{\circlearrowleft}$ together imply monotonicity of $\iteration$ and thereby $ \iteration (F)\preceq \iteration (\sup_\preceq P)$ for all $F \in P$. Therefore, it suffices to show $\iteration (\sup_\preceq P) \preceq \sup_\preceq \lbrace \iteration (F) \ | \ F \in P \rbrace$.
    Let $g\in G$ and $e \in \iteration (\sup_\preceq P)(g)$. We now show $e \in \ \uparrow \left( \sup_\preceq \lbrace \iteration (F) \ | \ F \in P \rbrace \right) (g)$. We focus on the more intricate case and assume $g \in G_d$. By definition of $\iteration$, there exist $e_{g'} \in \left( \sup_\preceq P\right) (g')$ for each $g'\in G$ with $g \move{u} g'$ such that $e = \sup \lbrace u^{\circlearrowleft} (e_{g'}) \ | \ g \move{u} g' \rbrace$. By definition of $\sup_\preceq$ there exist $F_{g'}\in P$ with $e_{g'} \in \ \uparrow F_{g'}(g')$. Since the set of positions is finite, so is $\lbrace F_{g'} \ | \ g \move{u} g' \rbrace$ and there exists an upper bound $F' \in P$ such that $e_{g'} \in \ \uparrow F'(g')$ for each successor $g'$ of $g$. Thus, $e \in \ \uparrow \mathsf{Iteration} (F') (g) \subseteq \ \uparrow \left( \sup_\preceq \lbrace \mathsf{Iteration} (F) \ | \ F \in P \rbrace \right) (g)$.

    With this, we can apply \autoref{lem:kleene} and conclude that a least fixed point, namely $\sup_\preceq \lbrace \iteration^i (\mathbf{0}) \ | \ i \in \N \rbrace$, exists. We now show, that the least fixed point is $\winamin$ by first arguing that $\winamin$ is a fixed point. For this, we introduce inductively defined sets $S(g)$ for $g\in G$ mimicking the process of the iteration of the while-loop without taking the minimum: 
    \begin{mathparpagebreakable}
    \inferrule{
      g \in G_a\\
      g \move{u} g'\\
      e' \in S(g')
    } {
      u^{\circlearrowleft} (e') \in S(g)
    }\and
    \inferrule{
      g \in G_d\\
    \forall g'. \  g \move{u} g' \longrightarrow e_{g'}\in S(g')
    } {
      \left( \sup \lbrace u^{\circlearrowleft}(e_{g'}) \ | \ g \move{u} g' \rbrace \right)  \in S(g)
    }
  \end{mathparpagebreakable}
    Then, $S (g) \subseteq \wina (g)$ for all $g\in G$ follows from a simple induction over the structure of $S$, utilizing the inductive characterization of attacker winning budgets as well as \autoref{lem:upwards-closure} and~\ref{lem:galoisproperties}. 
    Further, $\wina(g) \subseteq \ \uparrow S (g)$ can be shown by a well-founded induction over the strategy-induced order introduced in the proof of \autoref{lem:inductive_wina} using \autoref{lem:galoisproperties}. 
    With \autoref{lem:upwards-closure} we can conclude $\wina(g) = \ \uparrow S (g)$ and thereby $\winamin = \Min \circ \ S$. 
    Further, the definition of $S$ directly implies that $\Min \circ \ S$ (and therefore $\winamin$) is a fixed point of $\iteration$.
    For any fixed point $F$ of $\iteration$, we have $S(g) \subseteq {\uparrow F(g)}$ for all $g \in G$ by induction over the structure of $S$. Thus, $\Min \circ \ S$ (and therefore $\winamin$) is the least fixed point of $\iteration$.

    Finally, we argue for termination by noting that there exists $i_g^e \in \N$ such that $e \in \iteration^{i_g^e}(\mathbf{0})(g)$ for each $g\in G$ and $e \in \winamin (g)$, since $\winamin = \sup_\preceq \lbrace \iteration^i (\mathbf{0}) \ | \ i \in \N \rbrace$. Thus, a fixed point is reached after at most $\smash{\displaystyle\max_{g \in G} \max_{e\in \winamin (g)} i_g^e}$ iterations of the while-loop. \qedhere
\end{proof}

\subsection{Isabelle/HOL Formalization}
\label{subsec:isabelle}

The proof of decidability of Galois energy games (\autoref{thm:decidable}) is formalized in Isabelle/HOL.
Isabelle~\cite{isabelle} is a generic interactive proof assistant supporting the formalization of mathematical theories. Theorems formalized in Isabelle undergo automated verification, ensuring that edge cases are not overlooked. While Isabelle supports multiple logical frameworks, we use the instantiation Isabelle/HOL based on Higher Order Logic (HOL). 
The Isabelle theories are available at \url{https://github.com/crmrtz/galois-energy-games}.

Our Isabelle formalization starts by providing the first formalization of this kind for any type of energy game. This includes the inductive characterization of attacker winning budgets (\autoref{lem:inductive_wina}) assuming energy-positional determinacy (\autoref{lem:pos-determinacy}). 
Fixing an arbitrary Galois energy game over a well-founded bounded join-semilattice we then formalize the function $\iteration$ introduced in the proof of \autoref{thm:decidable}. 
We formalize the correctness of \autoref{alg:game-algorithm} by showing that the minimal attacker winning budgets are the least fixed point of $\iteration$ (correctness) and show termination.
Thus, we have formalized the decidability of such Galois energy games. 
Finally, we formalize Galois energy games over vectors of (extended) naturals with the component-wise order and conclude decidability, particularly for energy games with vector addition and taking minima of components only. 

\begin{table}[h]
    \begin{adjustbox}{max width=\textwidth, center}
    \centering
    \begin{tabular}{c|c|c|c}
         & Result & Formalization link & Proof adjustments\\
         \hline
        \autoref{lem:pos-determinacy} & Energy-positional determinancy & - \\
        \autoref{lem:upwards-closure} & Upwards-closure of $\wina$& \href{https://github.com/crmrtz/galois-energy-games/blob/db18c53dae10f826cafaff11899e9ebfb274b589/isabelle-theories/Energy_Game.thy#L685}{\texttt{upward\_closure\_wb\_nonpos}} \\
        \autoref{lem:inductive_wina} & Inductive characterization of $\wina$& \href{https://github.com/crmrtz/galois-energy-games/blob/db18c53dae10f826cafaff11899e9ebfb274b589/isabelle-theories/Energy_Game.thy#L2471}{\texttt{inductive\_winning\_budget}} & assuming \autoref{lem:pos-determinacy} \\
        \autoref{lem:galoisproperties} & Galois connection equivalences & \href{https://github.com/crmrtz/galois-energy-games/blob/db18c53dae10f826cafaff11899e9ebfb274b589/isabelle-theories/Galois_Energy_Game.thy#L61}{\texttt{galois\_properties}} & only $1. \longrightarrow 2.$\\
        \autoref{lem:kleene} & Kleene's fixed point theorem & by Yamada and Dubut~\cite{kleene-fp-AFP} \\
        \multirow{2}{*}{\autoref{thm:decidable}} & \multirow{2}{*}{Decidability}  & Correctness: \href{https://github.com/crmrtz/galois-energy-games/blob/db18c53dae10f826cafaff11899e9ebfb274b589/isabelle-theories/Decidability.thy#L3674}{\texttt{a\_win\_min\_is\_lfp}} & \multirow{2}{*}{assuming \autoref{lem:pos-determinacy}} \\ 
        & & Termination: \href{https://github.com/crmrtz/galois-energy-games/blob/db18c53dae10f826cafaff11899e9ebfb274b589/isabelle-theories/Decidability.thy#L3455} {\texttt{finite\_iterations}} & 
    \end{tabular}
    \end{adjustbox}
    \caption{Overview of the formalization}
    \label{table:formalization}
\end{table}
\autoref{table:formalization} provides an overview of all previously discussed lemmas and theorems, along with information on which ones have been formalized and links to their formalization.

\section{Instances of Galois Energy Games}
\label{sec:instances}

Let us discuss how Galois energy games and the algorithm can be used to solve quantitative problems by instantiating energy orders and updates over vectors of naturals.

For the illustration, we will take an extremely classical problem and two recent ones.
In \autoref{subsec:shortest-path}, we show how shortest paths problems on graphs are encoded as single-player energy games and relate the latter to vector addition systems with states (VASS).
As more recent problems, \autoref{subsec:instantiations} will discuss the declining energy games of Bisping~\cite{bens-algo} and the multi-weighted reachability games of Brihaye and Goeminne~\cite{mult-reach-games}, in a generalized version.

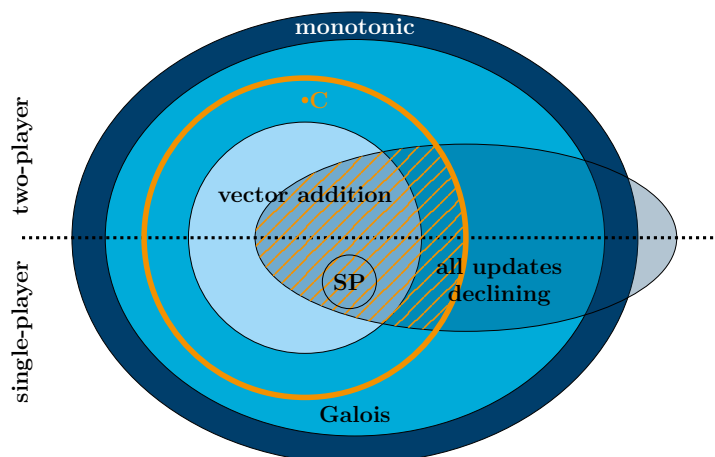
\begin{figure}[h]
\begin{adjustbox}{center, trim={0mm 2mm 0mm 0mm}}
\adjustbox{max width=0.8\textwidth, max height=6cm}{
\centering
    \begin{tikzpicture}
    \begin{scope}[
    mytext/.style={text opacity=1,font=\Large\bfseries}]

        \draw[fill=blue1, draw = black] (0,0) ellipse (5.1cm and 4.1cm);
        \draw[fill=blue2, draw = black] (0,0) ellipse (4.5cm and 3.6cm);
        \draw[fill=blue3, draw = black,name path=circle 1] (-0.9,0) circle (2.1);
        \draw[fill=blue1, draw = black,name path=circle 2, fill opacity=0.3] (2.0,0) ellipse (3.8cm and 1.7cm);
        \draw[black,ultra thick, dotted] (-6,0)--(6.5,0);

        \pgftransparencygroup
       \clip (1.8,0) (2.0,0) ellipse (3.8cm and 1.7cm);
       \filldraw [orange2, pattern={Lines[distance=2mm, angle=45, line width=0.3mm]}, pattern color=orange2](-0.9,0) circle (2.9);
        \endpgftransparencygroup
        \draw[orange2, line width=1mm](-0.9,0) circle (2.9);

        \node[mytext, white] at (0,3.85) (mono) {monotonic};
        \node[mytext] at (0,-3.2) (G) {Galois};
        \node[mytext, align=center] at (2.6,-0.8) (decl) {all updates\\ declining};
        \node[mytext] at (-0.9,0.8) (E) {vector addition};
        \node[mytext, rotate=90] at (-6,1.5) (2) {two-player};
        \node[mytext, rotate=90] at (-6,-1.7) (2) {single-player};
        \node[mytext, draw=black, circle] at (-0.1,-0.8) (SP) {SP};
        \node[mytext, orange2] at (-0.65,2.5) (C) {C};
        \draw[fill=orange2, draw=orange2] (-0.9,2.5) circle (0.05);
    \end{scope}
    \end{tikzpicture} 
}
\end{adjustbox}
    \caption{Energy games over vectors of naturals with the component-wise order}
    \label{fig:venn-games}
\end{figure}

\autoref{fig:venn-games} lays out the kinds of energy games to be discussed: 
The yellow circle includes all energy games with updates that mix positive, negative, and $\min$-updates, in particular, our example energy game of the researcher getting coffee placed as yellow \textbf{C}.
The crosshatched area includes the “declining energy games” discussed by Bisping~\cite{bens-algo}, where the subclass with vector addition only corresponds to the games discussed by Brihaye and Goeminne~\cite{mult-reach-games}. 
The bottom half-circle of vector addition contains energy games solving the coverability problem on VASS and \textbf{SP} in that set refers to the single-player energy games corresponding to shortest path problems for the one-dimensional case as will be discussed in \autoref{subsec:shortest-path} 

Note that the class of monotonic energy games is decidable.
For this, Lemke~\cite{Lemke2024} generalizes \autoref{alg:game-algorithm} by calculating sets of the form $\Min \lbrace e \in \E \ | \ e' \leq u(e) \rbrace$ if a unique minimum does not exist, which adds another level of combinatorial explosion.
For details on the boundaries of decidability for monotonic games, we refer to Abdulla et al.~\cite{monotonic-games}.
We focus our attention on Galois energy games, which generalize previous results in a simple yet concrete manner.

\subsection{Singe-Player Instances}\label{subsec:shortest-path}

\emph{Single-player} energy games are energy games, where all defender positions are deadlocks and, therefore, only the attacker makes choices.

\paragraph*{Shortest Paths and Distances}
Shortest path and shortest distance problems~\cite{mohri2002semiringShortest} are classical quantitative problems.
These can be seen as instances of “single-player” energy games, where the minimal attacker winning budgets correspond to a kind of \emph{shortest distances} from a position to any defender position.
Let us illustrate this for the one-dimensional case:

\begin{proposition}
    \label{prop:shortest-distance}
    Assume we have a graph $(V, E)$ with $E \subseteq V \times \N \times V$,
    and ask for the length of shortest paths from some source node $s \in V$ to target node $t \in V$.
    Consider the energy game $(V \cup \{\bot\}, V, \move{ \ })$ over $(\N, \leq)$ with $v \move{-w\ } v'$ if $(v,w,v') \in E$ and $t \move{\ 0\ } \bot$.
    Then, $e \in \winamin(s)$ on the derived game precisely if $e$ is the shortest distance from $s$ to $t$ in the original graph.
\end{proposition}
On this instance, our \autoref{alg:game-algorithm} boils down to a single-destination variant of the Bellman--Ford algorithm for shortest paths.

\begin{remark}
   Mixing positive and negative weights, the construction of \autoref{prop:shortest-distance} encounters differences to classical shortest-path solutions, which assume commutativity of edge costs.
    For example, in the graphs $s \xrightarrow{\mathmakebox[1em]{2}} \circ \xrightarrow{\mathmakebox[1em]{-1}} t$ and $s \xrightarrow{\mathmakebox[1em]{-1}} \circ \xrightarrow{\mathmakebox[1em]{2}} t$, Bellman--Ford will consider both $s$-$t$-paths to have length~$1$.
   However, the energy game interpretation leads to a different result:
   Moving $\move{\mathmakebox[1em]{-2}}\move{\mathmakebox[1em]{1}\vphantom{-2}}$ requires an attacker budget of $2$, while $\move{\mathmakebox[1em]{1}\vphantom{-2}}\move{\mathmakebox[1em]{-2}}$ only needs~$1$.
   In some cases the energy interpretation is more accurate in reality:
   If one has an empty tank, it makes a difference whether one can refuel now or only later.

   Also, conventional shortest-path algorithms struggle with negative cycles, which correspond to positive cycles in our games.
   In the energy interpretation, we easily determine the cost to reach a target even with an infinite charging cycle:
   The cost of getting to the cycle.  
\end{remark}

\paragraph*{Vector Addition Systems with States}

The single-player version of a multi-dimensional energy game over $\N^n$ with the component-wise order $\leq_c$ where edges are labeled by functions adding a vector in $\Z^n$ can be used to answer the \emph{coverability problem}~\cite{kunnemann2023coverability} on VASS.
In VASS, this problem asks whether a target position can be reached from a given initial configuration (position and energy) with at least a fixed energy.

\begin{proposition}
    \label{prop:coverability-VASS}
    Let $n\in \N$. Assume we have a VASS $(Q,T)$ with $T\subseteq Q \times \Z^n \times Q$ with a fixed initial configuration $(s,e_s) \in Q \times \N^n$ and a target configuration $(t,e_t) \in Q \times \N^n$. 
    Consider the energy game $(Q \cup \{\bot\}, Q, \move{ \ })$ over $(\N^n, \leq_c)$ with $v \move{\ w\ } v'$ if $(v,w,v') \in T$ and $t \move{\ -e_t\ } \bot$.
    Then, $e_s \in \wina(s)$ on the derived game precisely if coverability in the original VASS holds.
\end{proposition}
The \emph{reachability problem}~\cite{Blondin-PSPACE} on VASS aks whether a configuration is reachable from another. Since the target energy has to be reached exactly, this is a stricter version of coverability. 
As in \autoref{prop:coverability-VASS}, 
this can be modeled as an energy game by labeling the edge $t \move{\quad} \bot$ with a function to check if the current energy matches the target (that is undefined if not).
However, since such a function lacks an upward-closed domain, our results are not applicable and we refer to existing literature~\cite{Blondin-PSPACE, czerwinski2025reachability, czerwinski2022reachability} instead.

\subsection{Energy Games over Vectors of Naturals}
\label{subsec:instantiations}

In recent research, energy games with reachability winning conditions have garnered some attention~\cite{bens-algo, mult-reach-games}.
We now demonstrate that we generalize their results. We do so with a syntactical approach to define a set of updates permitted in Galois energy games over vectors of naturals. 
We introduce the set $\mathcal{U}_n$ of possible updates that allow changing entries in two ways: By adding an integer or replacing it with the minimum of components. Such updates can be represented as vectors where each entry contains the information how the respective component is updated. 
First fix a set of energies: 
The set of vectors of extended naturals $\N_\infty^n$ of dimension $n \in \N$ with the component-wise order $\leq_c$.

\begin{definition}
    Let $n\in \N$ and $\mathcal{U}_n := \left(\Z \cup 2^{\lbrace0,..., n-1\rbrace} \setminus \lbrace \varnothing \rbrace \right) ^n$. Then, we interpret $u= (u_0, ..., u_{n-1} ) \in \mathcal{U}_n$ as a partial function $u: \N_\infty^n \rightarrow \N_\infty^n$ by setting
    \begin{equation*}
        u ( e_0, ..., e_{n-1}) := (e_0', ... , e_{n-1}') \text{ with } e_i' := \begin{cases}
        e_i + z, &\text{if } u_i = z \land -z \leq e_i \\
        \min_{k \in D} e_k, &\text{if } u_i = D \subseteq \lbrace 0, ..., n-1 \rbrace
    \end{cases}
    \end{equation*}
    if all $e_i'$ are defined.
\end{definition}
As a consequence of \autoref{thm:decidable}, we obtain the following.
\begin{corollary}\label{col:U_nGaloisgames}
    Let $n\in \N$ and let $\mathcal{G}=(G, G_a, \move{ \ } )$ be an energy game over $(\N_\infty^n, \leq_c)$ with updates in $\mathcal{U}_n$. Then, $\mathcal{G}$ is a Galois energy game. Further, if the set of positions is finite, the game is decidable.\footnote{This is formalized using \autoref{rem:composition} to establish a subclass relationship, i.e.\ a  \href{https://github.com/crmrtz/galois-energy-games/blob/db18c53dae10f826cafaff11899e9ebfb274b589/isabelle-theories/Natural_Galois_Energy_Game.thy\#L105}{\texttt{sublocale}} relationship.} 
\end{corollary}

\begin{proof}
We argue that $\mathcal{G}$ is a Galois energy game: 
By definition, all updates in $\mathcal{U}_n$ are monotonic with an upward-closed domain. 
Generalizing the inversion in \cite{bens-algo}, we show that ${u^\circlearrowleft (e') = \min \lbrace e \ | \ e' \leq u(e) \rbrace}$ exists for all $u \in \mathcal{U}_n$ and $e' \in \E$ by giving a computable calculation. 
 
    Let $u$ be an update in $\mathcal{G}$, i.e.\ $u=(u_0, ..., u_{n-1}) \in \mathcal{U}_n$. For $e'=(e_0', ..., e_{n-1}') \in \N_\infty^n$ we set $u^\circlearrowleft (e') $ such that the $i$-th component is the maximum of 
        \begin{itemize}
        \item $e_i' -z$, if this is not negative and the update adds $z$ in the $i$-th component,
        \item $e_j'$, if the $i$-th component is used when taking the minimum resulting in the $j$-th component, 
        \item $0$ else.
    \end{itemize}
    That is
    \begin{equation*}
            u^\circlearrowleft (e')_i := \max ( 
            \lbrace e_i' - z \mid u_i = + z \in +\Z \land z \leq e_i' \rbrace 
            \cup 
        \lbrace e_j' \mid i \in u_j \subseteq \lbrace 0, ..., n-1\rbrace \rbrace 
            \cup \lbrace 0 \rbrace ). \qedhere
    \end{equation*}
\end{proof}

\begin{remark}
    All our arguments for $\mathcal{U}_n$ still hold, if we allow multiplication with a nonzero natural. When calculating $u^\circlearrowleft (e')_i$ we then have to add $ \ceil[\big]{\frac{e_i '}{m}}$ to the maximum, if the update multiplies the $i$-th component with $m$.
\end{remark}

\begin{example}\label{exm:composition}
    Note that all occurring updates in \autoref{exm:running-example} aside from the loop at $\mathsf{CoffeeMaker}$ are elements of $\mathcal{U}_4$. That loop can be replaced by adding a $\mathsf{Brew}$-position such that $\mathsf{CoffeeMaker} \move{+1 \, \mathsf{Shots} \,} \mathsf{Brew} \move{\mathsf{Shots \leftarrow \min \lbrace \mathsf{Shots}, \mathsf{Cups}\rbrace}} \mathsf{CoffeeMaker}$. 
\end{example}
Actually, the adjustment of the game graph described in \autoref{exm:composition} is not necessary to prove the game to be a Galois energy game given that Galois connections compose.

\begin{lemma}\label{lem:composition}
    Let $u_1, u_2 : \E \rightarrow \E$ be partial functions with upward-closed domains such that $u_i^\circlearrowleft$ and $u_i$ forms a Galois connection between $\E$ and $\dom (u_i)$ for $i\in \lbrace 1,2\rbrace$.
    Then, $(u_2 \circ u_1 )^\circlearrowleft := u_1^\circlearrowleft \circ u_2^\circlearrowleft$ and $u_2 \circ u_1$ form a Galois connection between $\E$ and the reverse image under $u_1$ of $\dom (u_2)$, i.e.\ $u_1^{-1} (\dom (u_2))$.\footnote{
    This is formalized as \href{https://github.com/crmrtz/galois-energy-games/blob/db18c53dae10f826cafaff11899e9ebfb274b589/isabelle-theories/Galois_Energy_Game.thy\#L212}{\texttt{galois\_composition}}.}
\end{lemma}

\begin{remark}\label{rem:composition}
    \autoref{lem:composition} implies that it suffices to study simple updates as generators and check that they are permitted in a Galois energy game, to prove decidability. 
    In particular, we may focus on updates in $\{-1,0,+1\}^n \subseteq \Z^n$ that add or subtract at most one in any component to understand all updates in $\Z^n \subseteq \mathcal{U}_n$. 
\end{remark}
Now we can consider instances where multi-weighted energy games with reachability conditions were considered. 
Bisping~\cite{bens-algo} introduces a game generalizing the bisimulation game to simultaneously decide all common notions of behavioral equivalence, i.e.\ in the linear-time-branching-time spectrum~\cite{vanGlabbeek}. In particular, he uses energy games over $(\N_\infty^n, \leq_c)$ with declining updates in $\mathcal{U}_n$. 
The decidability of those games is a direct implication of \autoref{col:U_nGaloisgames}. 

The multi-weighted reachability games described by Brihaye and Goeminne~\cite{mult-reach-games} can easily be transformed to energy games with declining updates in $\Z^n$ by utilizing the construction from \autoref{rem:reachability}. 
When considering the component-wise order, decidability follows directly from \autoref{col:U_nGaloisgames}.
When considering the lexicographical order $\leq_l$,
decidability can be concluded using \autoref{thm:decidable}.

\subsection{Implementations}

\autoref{alg:game-algorithm} has been implemented twice in open source software:
Once, in the \emph{Linear-Time–Branching-Time Spectroscope}, a web tool to decide spectra of behavioral equivalence~\cite{bens-algo}, and once in \emph{gpuequiv}, a GPU-accelerated version thereof~\cite{vogel2024energyGamesWebGPU}.
Both work on energy games of arbitrary-dimensional vectors of naturals.

The \emph{Spectroscope} is implemented in Scala.js, powering the tool on \url{https://equiv.io/}.
Its energy game class supports positive, negative, and $\min$-updates.
\autoref{exm:running-example} is reproduced as a unit test in
\href{https://github.com/benkeks/equivalence-fiddle/blob/main/shared/src/test/scala-2.12/io/equiv/eqfiddle/game/GaloisEnergyGameTests.scala}{\texttt{io.equiv.eqfiddle.game.GaloisEnergyGameTests}}.

\emph{gpuequiv}~\cite{vogel2024energyGamesWebGPU} uses Rust and the WebGPU system to solve declining energy games on the GPU.
The core is an implementation of our algorithm split up across several shaders in order to exploit GPU parallelism in \href{https://github.com/Gobbel2000/gpuequiv/blob/master/src/energygame.rs}{\texttt{gpuequiv::energygame::EnergyGame}}.
For the shader implementation, it is particularly important that the buffer size for the Pareto fronts can be bounded, as will be discussed in \autoref{sec:complexity}.
(In order to conserve memory, the shader will initially reserve space for fronts of up to 64 entries during updates, and must retry positions with more space whenever the allocated space does not suffice.)
So far, gpuequiv only supports declining energy games, for which it includes several tests.
According to benchmarking~\cite[Chapter 5]{vogel2024energyGamesWebGPU}, \emph{gpuequiv} tends to solve the games 10 to 20 times faster than the Scala implementation (even if the latter is compiled and run through JVM instead of JS).

Both implementations increase our confidence that the approach works well in practical scenarios,
and hopefully can be of help to others who want to solve similar problems.

\section{Complexity}
\label{sec:complexity}

The complexity of deciding Galois energy games varies based on the energy order and the updates. 
Since deciding multi-weighted reachability games is PSPACE-complete as stated by Brihaye and Goeminne~\cite{mult-reach-games}, deciding energy games over vectors of naturals with the component-wise order and only vector-subtraction is PSPACE-hard:

\begin{proposition}
    Deciding Galois energy games over $n$-dimensional vectors of naturals with the component-wise order is PSPACE-hard with respect to $n$.
\end{proposition}
In this section, we give more fine-grained upper complexity bounds for the generic algorithm and for the instantiations according to the previous section.

The running time of \autoref{alg:game-algorithm} depends on factors beyond input size, such as the occurring updates which influence the number of iterations needed. To analyze the influence of the order on complexity, we introduce functions that will help establish upper bounds for attacker winning budgets, which are then used in our complexity results.

\begin{definition}
    For $e\in \E$ let $\mathsf{antichain}_{\leq e}$ be the set of antichains in $\lbrace e' \in \E \ | \ e' \leq e \rbrace$. Then, 
    $\h (e)$ is the cardinality of $\lbrace e' \in \E \ | \ e' \leq e \rbrace$ and 
    $\w(e)$ is the maximal cardinality of elements in $\mathsf{antichain}_{\leq e}$, i.e.\ $\w (e) := \textstyle{\sup_\N} \lbrace \relSize{A} \ | \ A \in \mathsf{antichain}_{\leq e}\rbrace$.
\end{definition}
Fixing even more parameters, we can state our complexity result.

\begin{theorem}[Complexity]\label{thm:complexity}
    Let $\mathcal{G}=(G, G_a, \move{ \ } )$ be a Galois energy game over $(\E, \leq)$. Further, let
    \begin{itemize}
        \item $o$ be the branching degree of the underlying graph,
        \item $t_{\sup}$ be an upper bound of the time to compute the supremum of two energies,
        \item $t_{\leq}$ be an upper bound of the time to compute the comparison of two energies,
        \item $t_{{\circlearrowleft}}$ be an upper bound of the time to compute $u^{\circlearrowleft} (e)$ for any update $u$ in $\mathcal{G}$ and $e \in \E$, and
        \item $e_\mathit{worst} := \sup \lbrace u_1^{\circlearrowleft} \circ ... \circ u_i^{\circlearrowleft} (\min \E) \ | \ i \in \lbrace 0, ..., \relSize{G} -1 \rbrace \  \land \ u_1, ..., u_i \text{ are updates in } \mathcal{G} \rbrace$ be the highest energy obtainable using up to $\relSize{G} -1$ reverse updates applied to the least element of $\E$.
    \end{itemize}   
    Then, \autoref{alg:game-algorithm} on input $\mathcal{G}$ terminates in a time in 
    \begin{align*}
        &\mathcal{O}\left(\relSize{G}^2 \cdot o \cdot \h(e_\mathit{worst}) \cdot \w(e_\mathit{worst})^2 \cdot (t_\circlearrowleft + t_{sup} + \w(e_\mathit{worst}) \cdot t_{\leq} )\right),\\
        &\text{and in }\mathcal{O}\left(\relSize{G}^2 \cdot o \cdot \w(e_\mathit{worst})^2 \cdot (t_\circlearrowleft + t_{sup} + \w(e_\mathit{worst}) \cdot t_{\leq} )\right) \text{, if }\mathcal{G} \text{ is declining.}
    \end{align*}
    In both cases, the output is calculated using space in $\mathcal{O} \left( \relSize{G} \cdot \w(e_\mathit{worst}) \cdot s_\E \right)$ where $s_\E$ is an upper bound for the space needed to store any energy.
\end{theorem}

For a detailed proof, we refer to the appendix. 
Compared to the approach by Bisping~\cite{bens-algo}, this result significantly improves complexity -- see \autoref{table:complexity}. The key argument leading to this improvement is the calculation of \autoref{code:defender-pos-update} of \autoref{alg:game-algorithm}, i.e.\ $\varname{win} [\varname{g}]$ for a defender position $\varname{g}$. Inspired by Brihaye and Goeminne~\cite{mult-reach-games} we apply the following procedure:

\begin{algorithm}[H]
  \Fn{$\varname{compute\_new\_win}(\mathcal{G}, \varname{old\_win}, \varname{g})$}{
    $\varname{new}[\varname{g}]:= \Min \E$

    \For{$g'$ with $\varname{g} \move{u} g'$\ }
    {

    $\varname{new}[\varname{g}] := \Min \lbrace \sup \lbrace e', u^{\circlearrowleft}(e_{g'}) \rbrace \ | \ e' \in \varname{new}[\varname{g}]\land  e_{g'} \in \varname{old\_win}[g'] \rbrace $
    
    }
    
    \KwRet{$\varname{new}[\varname{g}]$}
  }
\end{algorithm}

\bigskip
\autoref{table:complexity} compares the complexity results we obtain by \autoref{thm:complexity} to results stated by Bisping~\cite{bens-algo} on declining energy games and Brihaye and Goeminne~\cite{mult-reach-games} on multi-weighted reachability games, respectively. Further, it states the complexity of \autoref{alg:game-algorithm} when applied to a shortest path problem as discussed in \autoref{prop:shortest-distance} and for energy games with updates in $\Z^n \subseteq \mathcal{U}_n$, i.e.\ integer vector addition, subsuming the complexity of solving coverability on VASS as seen in \autoref{prop:coverability-VASS}.

To do so, we first give upper bounds for some variables: In all cases, we over-approximate $o$ with $\relSize{G}$ and assume $t_{\leq_l}, t_{\leq_c}, t_{\sup} \in \mathcal{O}(n)$. We use $w$ as an upper bound of the absolute value of the integers being added to a component at any edge.\footnote{To get a fair comparison we consider the complexity of declining energy games as formally defined by Bisping~\cite{bens-algo} where an update may subtract at most one from a component, i.e.\ $w=1$.} Then, we have $e_\mathit{worst}= w \cdot (\relSize{G}, \ldots, \relSize{G})$ and $\h (e_\mathit{worst}) \in \mathcal{O}\left((w \cdot \relSize{G})^n \right)$. Other upper bounds are stated in \autoref{table:complexity}, where all entries are to be understood in big $\mathcal{O}$-notation.

\begin{table}[h]
    \centering
    \begin{adjustbox}{max width=1.1\textwidth, center}
    \begin{tabular}{c|lllll}
        Game model & $t_{{\circlearrowleft}}$ &  $\w (e_\mathit{worst})$ & time complexity & complexity in cited work
        \\
        \hline
        Declining energies \cite{bens-algo} & $n^2$ &$\relSize{G}^{n-1}$ & $n^2 \cdot \relSize{G} ^{3n} $ & $\relSize{\move{ \ }} \cdot \relSize{G}^{n} \cdot (o+ \relSize{G}^{(n-1) o})$
        \\
        Multi-reachability, component-wise $\leq_c$ \cite{mult-reach-games} & $n$ & $w^{n-1} \cdot \relSize{G}^{n-1}$ & $n \cdot w^{3(n-1)} \cdot  \relSize{G} ^{3n} $ & $n^4 \cdot w^{4n} \cdot \relSize{G}^{4n +1}$
        \\
        Multi-reachability, lexicographic $\leq_l$ \cite{mult-reach-games} & $n$ & $1$ & $n \cdot \relSize{G} ^3 $ & $n \cdot \relSize{G} ^3$
        \\
        Shortest paths (with $\N$-edges)  & $1$ & $1$ & $\relSize{G}^{3}$ & .
        \\
        Integer vector addition, component-wise $\leq_c$& $n$ & $w^{n-1} \cdot \relSize{G}^{n-1}$ & $n \cdot w^{4n-3} \cdot  \relSize{G} ^{4n} $ & .
    \end{tabular}
    \end{adjustbox}
    \medskip
    \caption{Application of \autoref{thm:complexity} to declining classes of energy games.}\label{table:complexity}
\end{table}

\section{How to Handle Other Winning Conditions}
\label{sec:winning-conditions}

Reachability winning conditions are fundamental to the analysis of games, both in their simplicity and their importance. In this section, we examine how other types of winning conditions relate to reachability (\autoref{def:winning}) studying related games. 

\subsection{Within the Scope: Captured Games}

\subparagraph*{Weak Upper Bounds:}
The defender winning if the energy level becomes undefined is another way of saying there is a lower bound for the energy (and the attacker is energy-bound). We now consider energy games with an additional \emph{upper bound}~\cite{automata-bounds, optimalBounds} of resources the attacker may hold.  
A \emph{weak upper bound} caps the energies.
In \autoref{exm:running-example}, the number of cups resembles a weak upper bound of the number of shots.
Similarly, a weak upper bound can be added to any energy game, where the minimum is monotonic, by adding a new dimension for each dimension that should be bounded and composing every update with the function taking the minimum of the dimension and its upper bound.
Then, \autoref{alg:game-algorithm} can be used in the case of both, a known and an unknown weak upper bound.

\subparagraph*{Generalized Reachability:}
\emph{Generalized reachability}~\cite{reachability-games-and-friends} refers to reachability games with multiple target sets where the attacker wins a play if each set has been visited at least once (without running out of energy). Combining the construction in \autoref{rem:reachability} and that for weak upper bounds, an energy game can model a generalized reachability game by encoding the target sets using additional dimensions that track the visited sets. For this construction, we add positions and edges such that the attacker may choose after reaching any target set, whether to continue or to move to the only deadlock defender position at the cost of one in every dimension corresponding to a target set. 
Note that this construction increases the dimension and thereby increases time complexity of our algorithm for computing minimal winning budgets exponentially in the number of target sets.

\subsection{At the Borders: Related but Distinct Games}

\subparagraph*{Strong Upper Bounds:}
We reconsider energy games with an upper bound of the attackers resources.  
Exceeding a \emph{strict upper bound}~\cite{optimalBounds} results in the defender winning.
Such an upper bound can be integrated similarly to a weak upper bound by adding dimensions for the bounded dimensions and instead of taking the minimum adding the difference between the energy and the updated energy to the bound dimension at every edge. (Inspired by Juhl et al.~\cite{optimalBounds} we apply the mapping $b_i \mapsto b_i + (e_i - u(e)_i)$ for each bounded dimension $i$ when $e \mapsto u(e)$ is calculated.)
However, if the updates are more complex than vector addition, this construction might not yield a Galois energy game -- the game might not even be monotonic. 

\subparagraph*{Safety Winning Conditions:}
Instead of aiming to reach a target set, the goal of \emph{safety}~\cite{monotonic-games} winning conditions is to avoid a set of unsafe positions. This corresponds to changing the winner of infinite plays to the attacker (if the energy level does not become undefined). 
Proving decidability then is symmetric to the proofs presented in this paper and can be achieved by adjusting the definitions and statements accordingly. In particular, this includes calculating maximal defender winning budgets instead and defining Galois energy games using \autoref{lem:galoisproperties}.\ref{lem:galoisMax}, i.e.\ as monotonic energy games where $u^\circlearrowright (e') := \max \lbrace e \in \E \uplus \lbrace \bot\rbrace \ | \ u(e) \leq e' \rbrace$ always exists and is computable. (Here $\bot$ represents undefined energies with $\bot < e$ for all $e \in \E$ and $u(\bot) := \bot$.)

\subparagraph*{Parity Games:}
An \emph{energy parity game}~\cite{abdulla-parity-on-integer, schewe2019parity} is an energy game with a ranking function that assigns each position a natural number (called priority) and has a finite image. Then, infinite plays where the energy level never becomes undefined are won by the attacker if and only if the lowest priority appearing infinitely often is even.
A stronger objective for the attacker would be to never visit more than $l$ positions with odd priority before visiting a position with a smaller even priority for a fixed $l\in \N$. Abdulla et al.~\cite{abdulla-parity-on-integer} showed that deciding games with that stronger objective in some cases is equivalent to deciding energy parity games. 
Such energy parity games with vector addition only can  be formulated as an energy game with safety winning conditions as described by Chatterjee et al.~\cite{strategy-synthesis} where dimensions corresponding to the odd priorities are added.
Whether such a procedure can be applied to other kinds of updates remains an open question.

\subsection{Beyond the Horizon: Games Outside the Framework}

\subparagraph*{Both-Bounded Energy Games:}
In \emph{both-bounded energy games}~\cite{kupferman2022energy} both players are energy-bound. This is modeled by assigning two updates to each edge, operating on pairs of current energies -- one for the attacker, one for the defender. Considering such games over vectors of naturals, undecidability arises if one player has one-dimensional energies and the other two, so our results cannot be transferred. 

\subparagraph*{Payoffs of Plays:}
\begin{sloppypar}
To compare infinite plays and decide a winner, different \emph{payoffs}~\cite{bouyer2018average} can be calculated. The \emph{mean-payoff} refers to the limit average gain per edge, i.e.\ $\limsup_{n\rightarrow \infty} \frac{1}{n} EL(\rho, e_0, n)$ for a play $\rho$ w.r.t.\ an initial energy $e_0$, while the \emph{total-playoff} is the limit energy level, i.e.\ $\limsup_{n\rightarrow \infty} EL(\rho, e_0, n)$, and the \emph{average-energy} refers to the limit average energy level, i.e.\ $\limsup_{n\rightarrow \infty} \frac{1}{n} \sum_{i=0}^n EL(\rho, e_0, i)$. Since \autoref{alg:game-algorithm} relies on the fact that it suffices to consider finite plays, our results cannot be transferred directly.
\end{sloppypar}

\section{Conclusion}
\label{sec:conclusion}

In this paper, we presented a simple algorithm for calculating minimal winning budgets, solving the (un)known initial credit problem for energy games over any well-founded bounded join-semilattice.
We generalized two recent approaches~\cite{bens-algo, mult-reach-games} and improved the running time.
Our results demonstrate that Galois connections offer a simple framework for constructing decidable energy games.
Compared to Abdulla et al.'s monotonic games~\cite{monotonic-games} or to Mohri's generalized shortest-distances~\cite{mohri2002semiringShortest}, our work is more concrete in its application.
For a relevant class of energy games over vectors of naturals, we offer an algorithm implementation for practical use.
Additionally, we provide an Isabelle/HOL formalization of the general decidability proof, ensuring confidence in our results.
By exploring various winning conditions and relating games to Galois energy games, we highlighted the theory's aplicability and limitations and identified areas for further research, such as using Galois connections to classify energy parity games.

\bibliography{galois-energy-games}

\appendix

\section*{Appendix}

This appendix contains the proofs of selected results from the main body of the paper.

\paragraph*{Energy-positional Determinacy}

\begingroup
\def\thelemma{\ref{lem:pos-determinacy}}
\begin{lemma}
    Let $\mathcal{G}$ be an energy game over $(\E, \leq)$.
    For all positions $g\in G$ and energies $e \in \E$ either the defender or the attacker wins $\mathcal{G}[g,e]$. Then, there exists an energy-positional winning strategy, i.e.\ a winning strategy that depends only on the current position and energy level. 
\end{lemma}
\addtocounter{lemma}{-1}
\endgroup

\begin{proof}
The (energy-positional) determinacy follows from the positional determinacy of parity games. We define a parity game $\mathcal{P_G} := (G_{\mathcal{G}},G_{\mathcal{G} 0}, E_{\mathcal{G}}, r_{\mathcal{G}})$ where 
    \begin{itemize}
        \item $G_{\mathcal{G}} := \lbrace (g,e) \ | \ g\in G, e \in \E\uplus \lbrace \bot \rbrace \rbrace$,
        \item $G_{\mathcal{G} 0} := G_{\mathcal{G}} \setminus G_{\mathcal{G} 1}$ with $ G_{\mathcal{G} 1} := \lbrace (g,e) \in G_{\mathcal{G}} \ | \ g\in G_a \lor e = \bot \rbrace$,
        \item $E_{\mathcal{G}} := \lbrace \left( (g,e) , (g', e') \right) \in G_{\mathcal{G}} \times G_{\mathcal{G}} \ | \ g \move{u} g' \land e \neq \bot \land u(e) = e' \rbrace$ and
        \item $r_{\mathcal{G}} : G_{\mathcal{G}} \rightarrow \mathbb{N} ,(g,e)\mapsto 0$.
    \end{itemize}  
This construction is similar to that of a configuration graph. Thereby, a play in $\mathcal{P_G}$ corresponds to a play in $\mathcal{G}$ and vice versa. 
Note that player~0 wins all infinite plays in $\mathcal{P_G}$ while the defender wins all infinite plays in $\mathcal{G}$.
By construction a position $(g,e)$ is a deadend in $\mathcal{P_G}$ if and only if $g$ is a deadend in $\mathcal{G}$ or $e=\bot$.
Hence, a play is won by the defender (resp. attacker) in $\mathcal{G}$ if and only if the corresponding play is won by player~0 (resp. player~1) in $\mathcal{P_G}$.

Let $e\in \E$ and $g \in G$.
Since parity games are positionally determined, there exists a positional winning strategy $s$ for player~0 or player~1 when starting in  $(g,e)$. Define $s_p: G^* G_p \rightarrow G$ by setting $s_p (g_0 ... g_n) := g'$ with $s  \left( g_n, EL(g_0 ... g_n, e) \right) = (g',e')$ if $g_n \in G_p$, where $G_p = G_d $ if player~0 wins and $G_p = G_a$ else. Note that this strategy is an energy-positional winning strategy.\qedhere

\end{proof}

\paragraph*{Complexity}

\begingroup
\def\thetheorem{\ref{thm:complexity}}
\begin{theorem}[Complexity]
    Let $\mathcal{G}=(G, G_a, \move{ \ } )$ be a Galois energy game over $(\E, \leq)$. Further, let
    \begin{itemize}
        \item $o$ be the branching degree of the underlying graph,
        \item $t_{\sup}$ be an upper bound of the time to compute the supremum of two energies,
        \item $t_{\leq}$ be an upper bound of the time to compute the comparison of two energies,
        \item $t_{{\circlearrowleft}}$ be an upper bound of the time to compute $u^{\circlearrowleft} (e)$ for any update $u$ in $\mathcal{G}$ and $e \in \E$, and
        \item $e_\mathit{worst} := \sup \lbrace u_1^{\circlearrowleft} \circ ... \circ u_i^{\circlearrowleft} (\min \E) \ | \ i \in \lbrace 0, ..., \relSize{G} -1 \rbrace \  \land \ u_1, ..., u_i \text{ are updates in } \mathcal{G} \rbrace$ be the highest energy obtainable using up to $\relSize{G} -1$ reverse updates applied to the least element of $\E$.
    \end{itemize}   
    Then, \autoref{alg:game-algorithm} on input $\mathcal{G}$ terminates in a time in 
    \begin{align*}
        &\mathcal{O}\left(\relSize{G}^2 \cdot o \cdot \h(e_\mathit{worst}) \cdot \w(e_\mathit{worst})^2 \cdot (t_\circlearrowleft + t_{sup} + \w(e_\mathit{worst}) \cdot t_{\leq} )\right),\\
        &\text{and in }\mathcal{O}\left(\relSize{G}^2 \cdot o \cdot \w(e_\mathit{worst})^2 \cdot (t_\circlearrowleft + t_{sup} + \w(e_\mathit{worst}) \cdot t_{\leq} )\right) \text{, if }\mathcal{G} \text{ is declining.}
    \end{align*}
    In both cases, the output is calculated using space in $\mathcal{O} \left( \relSize{G} \cdot \w(e_\mathit{worst}) \cdot s_\E \right)$ where $s_\E$ is an upper bound for the space needed to store any energy.
\end{theorem}
\addtocounter{theorem}{-1}
\endgroup

\begin{proof}
\begin{sloppypar}
    If all successors of a defender position $g_d$ were previously assigned a non-empty set, then $g_d$ will be assigned a non-empty set in the next iteration. Similarly, if any successor of an attacker position $g_a$ was assigned a non-empty set, then $g_a$ will be assigned a non-empty set next. 
    After $\relSize{G}$ iterations the information, which positions are winnable for the attacker, has traveled from defender deadlocks to all other winnable positions, i.e.\ $\forall g \in G. \ \wina (g) \neq \varnothing \longleftrightarrow \mathsf{Iteration}^{\relSize{G}}(\mathbf{0})(g)  \neq \varnothing$. 
    This allows us to calculate an upper bound for the energies being assigned to positions during the run of the algorithm, i.e.\ $e_\mathit{worst}$. 

    If  $\mathcal{G}$ is declining, then cycles do not benefit the attacker and a fixed point is reached in $\mathcal{O}(\relSize{G})$ iterations. In both, declining and non-declining games, only sets in $\mathsf{antichain}_{\leq e_\mathit{worst}}$ are assigned to positions as candidates for minimal attacker winning budgets. 
    Each set in $\mathsf{antichain}_{\leq e_\mathit{worst}}$ can only truly be updated at most as many times as the maximal length of chains in $\mathsf{antichain}_{\leq e_\mathit{worst}}$ w.r.t.\ inclusion of upward-closures is long, i.e.\ at most ${\textstyle{\sup_\N} \lbrace i \in \N \ | \ \exists A_1,... ,A_{i} \in \mathsf{antichain}_{\leq e_\mathit{worst}}. \ \forall j \in \lbrace 1, ..., i-1 \rbrace. \  A_j \subsetneq \ \uparrow A_{j+1}\rbrace =\h (e_\mathit{worst})}$ times.
    Thus, in non-declining energy games the algorithm terminates after at most $1+\relSize{G} + \relSize{G} \cdot \h (e_\mathit{worst}) \in \mathcal{O} (\relSize{G} \cdot \h (e_\mathit{worst}) ) $ iterations of the while-loop. 

    The computing time needed for operations outside the while-loop is negligible. We now focus on the time needed to execute one iteration, where \autoref{code:attacker-pos-update}'s time is dominated by that of \autoref{code:defender-pos-update}.     
    The following procedure inspired by Brihaye and Goeminne~\cite{mult-reach-games} calculates  \autoref{code:defender-pos-update}, i.e.\ $\varname{win} [\varname{g}]$ for a defender position $\varname{g}$. 
    
    \begin{algorithm}[H]
  \Fn{$\varname{compute\_new\_win}(\mathcal{G}, \varname{old\_win}, \varname{g})$}{
    $\varname{new}[\varname{g}]:= \Min \E$

    \For{$g'$ with $\varname{g} \move{u} g'$\ \label{code:branching}}
    {

    $\varname{new}[\varname{g}] := \Min \lbrace \sup \lbrace e', u^{\circlearrowleft}(e_{g'}) \rbrace \ | \ e' \in \varname{new}[\varname{g}]\land  e_{g'} \in \varname{old\_win}[g'] \rbrace $\label{code:defender-one-pos}
    
    }
    
    \KwRet{$\varname{new}[\varname{g}]$}
  }
  \end{algorithm}

    Note that the repeated application of $\Min$ ensures that $\w(e_\mathit{worst})$ is an upper bound for the size of $\varname{new}[\varname{g}]$ as well as $\varname{old\_win}[g']$ for all $g'\in G$. Therefore, \autoref{code:defender-one-pos} in the procedure can be calculated in $\mathcal{O} \left( \w(e_\mathit{worst})^2 \cdot (t_{\circlearrowleft} + t_{\sup} + \w(e_\mathit{worst}) \cdot t_{\leq}) \right) $. 
    The for-loop in \autoref{code:branching} adds a factor of the branching degree $o$, while \autoref{code:begin-iteration} adds $\relSize{G}$.
    This yields a running time of \autoref{alg:game-algorithm} in $ \mathcal{O}\left(\relSize{G}^2 \cdot o  \cdot \h(e_\mathit{worst}) \cdot \w(e_\mathit{worst})^2 \cdot (t_\circlearrowleft + t_{sup} + \w(e_\mathit{worst}) \cdot t_{\leq} )\right)$ where $\h(e_\mathit{worst})$ may be omitted in the case of declining energy games.

    Since each $\varname{win}[g]$ calculated contains at most $\w(e_\mathit{worst})$ elements, the space needed to compute the output is proportional the space of $ \relSize{G} \cdot \w(e_\mathit{worst})$ energies. \qedhere
\end{sloppypar}
\end{proof}

\end{document}